\documentclass[a4paper,UKenglish]{article}
\usepackage{amsfonts}
\usepackage{amssymb,mathtools}
\usepackage{amsmath}
\usepackage{authblk}
\usepackage{cite}
\usepackage{comment}
\usepackage[title]{appendix}
\usepackage{color}
\usepackage[linesnumbered,lined,ruled,vlined,noend]{algorithm2e}

\usepackage{todonotes}
\usepackage{apxproof}

\newcommand{\set}[1]{\{#1\}}
\newcommand{\inset}[2]{\{#1 \mid #2\}}
\newcommand{\name}[1]{\emph{#1}}

\newcommand{\size}[1]{\left|#1\right|}
\newcommand{\order}[1]{O(#1)}

\newtheorem{lemma}{Lemma}
\newtheorem{theorem}[lemma]{Theorem}
\newtheorem{problem}{Problem}

\newcommand{\cand}[1]{C\left(#1\right)}

\newcommand{\cin}[1]{C_{\rm{in}}{\left(#1\right)}}
\newcommand{\cout}[1]{C_{\rm{out}}{\left(#1\right)}}

\newcommand{\sig}[1]{\mathcal{#1}}

\newcommand{\DistM}[2][]{D^{(1)}_{#1}(#2)}
\newcommand{\TadM}[2][]{D^{(2)}_{#1}(#2)}
\newcommand{\DistMS}[2][]{D^{\footnotesize{(}3\footnotesize{)}}_{#1}(#2)}

\SetKwFunction{ELG}{EBG}
\newcommand{\ELGISname}{\texttt{EBG-IS}}
\SetKwFunction{ELGIS}{EBG-IS}
\newcommand{\ELGSname}{\texttt{EBG-S}}
\SetKwFunction{ELGS}{EBG-S}
\SetKwFunction{RecELG}{RecEBG}
\SetKwFunction{RecELGIS}{RecEBG-IS}

\SetKwFunction{NextC}{NextC}%
\SetKwFunction{UpdateCand}{UpdateCand}%
\SetKwFunction{UpdateOne}{Update1}%
\SetKwFunction{UpdateTwo}{Update2}%
\SetKwFunction{UpdateThree}{Update3}%

\makeatletter
\@ifpackageloaded{todonotes}{

}{}
\makeatother

\makeatletter
\@ifpackageloaded{algorithm2e}{
\SetKwProg{Fn}{Function}{}{}
\SetKwProg{Procedure}{Procedure}{}{}
\SetKwProg{Subprocedure}{Subprocedure}{}{}
\SetKwComment{tcc}{//}{}
\SetKwFunction{Output}{Output}%
\SetKw{Continue}{continue} \SetKwInOut{AlgInput}{Input}
\SetKwInOut{AlgOutput}{Output}
\SetKwInOut{AlgPrecondition}{Pre-conditions}
\SetKwInOut{AlgInvariant}{Invariants}
}{}
\makeatother

\title{Efficient Enumeration of Subgraphs and Induced Subgraphs with Bounded Girth}
\author[1]{Kazuhiro Kurita}
\author[2]{Kunihiro Wasa}
\author[2]{Alessio Conte}
\author[1]{Hiroki Arimura}
\author[2]{Takeaki Uno}

\affil[1]{IST, Hokkaido University, Sapporo, Japan\\
  \texttt{\{k-kurita, arim\}@ist.hokudai.ac.jp}}
\affil[2]{National Institute of Informatics, Tokyo, Japan\\
  \texttt{\{wasa, conte, uno\}@nii.ac.jp}}

\date{}
\begin{document}

\maketitle
\begin{abstract}
  The girth of a graph is the length of its shortest cycle.
Due to its relevance in graph theory, network analysis and practical fields such as distributed computing, girth-related problems have been object of attention in both past and recent literature.
In this paper, we consider the problem of listing connected subgraphs with bounded girth. As a large girth is index of sparsity, this allows to extract sparse structures from the input graph. 
We propose two algorithms, for enumerating respectively vertex induced subgraphs and edge induced subgraphs with bounded girth, both running in $O(n)$ amortized time per solution and using $O(n^3)$ space. 
Furthermore, the algorithms can be easily adapted to relax the connectivity requirement and to deal with weighted graphs.
As a byproduct, the second algorithm  can be used to answer the well known question of finding the densest $n$-vertex graph(s) of girth $k$.

\end{abstract}
%


\section{Introduction}

%

We consider the problem of finding all subgraphs and induced subgraphs with girth at least $k$ of a graph.
The girth is a measure of sparsity, as graphs with large girth are inherently sparse.
This corresponds to finding \emph{sparse} substructures of the given graph, a problem that was considered under several 
forms ~\cite{Jonson:Yannakakis:IPL:1988,Alessio:Kante:COCOON:2017} and has applications in network analysis.
In particular, this problem generalizes two well studied problems, i.e., listing all subtrees and 
induced subtrees~\cite{Read:Tarjan:Networks:1975,Shioura:Tamura:SICOMP:1997,Ferreira:Grossi:ESA:2011,Wasa:Arimura:Uno:ISAAC:2014}. 
Indeed, any graph with girth larger than $n$ may not contain a cycle, i.e., it is a tree, or a forest.

A \name{subgraph enumeration problem}, given a graph $G$ and some constraint $\sig R$, consists in outputting all the subgraphs satisfying $\sig R$ without duplicates. 
The efficiency of enumeration algorithms is often measured with respect to both the size of the input and that of the output, i.e., the number of solutions: 
an enumeration algorithm is called an
\name{amortized polynomial time algorithm}
if it runs in $\order{M\cdot poly(N)}$ time, where $N$ is the input size and $M$ is the number of solutions. 
Furthermore, the algorithm is said to have polynomial \emph{delay} if the maximum time elapsed between two consecutive outputs is polynomial. 

In this paper, we present two amortized polynomial time algorithms
for enumerating subgraphs of girth at least $k$. The first, \ELGIS, enumerates \textit{induced} subgraphs,
while the second, \ELGS, enumerates \textit{edge} subgraphs (also simply called subgraphs).
%
Both \ELGIS and \ELGS run in $\order{n\size{\sig S}}$ time 
using $\order{n^3}$ space, where $n$ is the number of nodes in $G$ and $\sig S$ is the set of all solutions.
%
%
%
The proposed algorithms will consider the enumeration of \textit{connected} subgraphs in simple graphs.
However, both algorithms can easily be applied to the enumeration of non-connected subgraphs, 
and to weighted graphs by trivial changes, with the same time and space complexity.
In these problems, 
the upper bound of the number of solutions are $\order{2^n}$ and $\order{2^m}$, respectively, where $m$ is the number of edges. 
Hence, the brute force algorithms are optimal if we evaluate the efficiency of algorithms only the input size. 
When we describe a more efficient algorithm, reducing amortized complexity is important~\cite{Kazuhiro:Kunihiro:arXiv:2018}.  
Indeed, our implementation of \ELGS \footnote{The implementation of \ELGS in the github repository: https://github.com/ikn-lab/EnumerationAlgorithms/tree/master/BoundedGirth/}
is almost $560$ times faster than the brute force algorithm when the input graph is a complete graph $K_8$ and girth is four.

While the problem of efficiently enumerating subgraphs with bounded girth has been considered for \emph{directed} graphs~\cite{Alessio:Kazuhiro:COCOA:2017}, to the best of our knowledge, there is no known efficient algorithm for the \emph{undirected} version of the problem.
\footnote{We remark that the techniques in~\cite{Alessio:Kazuhiro:COCOA:2017} do not extend to undirected graphs, thus motivating a separate study. 
In directed graphs, a $u$-$v$ path and a $v$-$u$ path are distinct. However, a $u$-$v$ path and a $v$-$u$ path may be same in undirected graphs. }

An early result on girth computation is the algorithm by Itai and Rodeh~\cite{Itai:Rodeh:SIAM:1978}, that finds the girth of a graph in $\order{nm}$ time. In more recent work, the problem was also solved in linear time for planar graphs~\cite{Chang:Lu:SIAM:2013}. 
However, the problem we consider involves computing the girth of many subgraphs, so relying on these algorithms is not efficient. 

A prominent question related to the girth is finding exactly how dense a graph of given girth can be:
the maximum number of edges in a $d$-regular graph with girth $k$ is bounded by the well known \emph{Moore bound}~\cite{Bollobas:BOOK:2004}, which Alon later proved to be tight on general graphs as well~\cite{Alon:Hoory:GC:2002}.
Erd{\H{o}}s conjectured that there exists a graph with $\Omega(n^{1 + 1/k})$ edges and girth $2k + 1$~\cite{Parter:ICALP:2014}. 
On the other hand, some have focused on giving practical lower bounds, i.e., finding ways to generate graphs of given girth as dense as possible~\cite{Lazebnik:Ustimenko:BAMS:1995,Chandran:SIMADM:2003}. We remark that our proposed algorithm \ELGS can match theory and practice: the densest $n$-vertex graph of girth $k$ can be found as a subgraph of the complete graph $K_n$. While this may not be practical for large values of $n$, it significantly improves upon the brute force approach by avoiding the generation of subgraphs with girth $< k$.


 \section{Preliminaries}
\label{sec:prelim}

Let $G = (V(G), E(G))$ be a simple undirected graph with no self-loops, with vertex set $V(G)$ and edge set $E(G) \subseteq V(G)\times V(G)$.
Two vertices $u$ and $v$ are \name{adjacent} (or neighbors) if there is an edge $e = \set{u,v} \in E(G)$ joining them. 
We call $e$ incident to $v$ and we denote the set of incident edges to $v$ $E(v)$. 
The set of neighbors of $u$ in $G$ is called its \name{neighborhood} and denoted by $N_G(u)$ and 
the size of $N_G(u)$ is called the \name{degree} of $u$ in $G$. 
Let $N_G[u] = N_G(u) \cup \set{u}$ be the closed neighborhood of $u$.
The \name{set of neighbors} of $U \subseteq V$ is defined as $N_G(U) = \bigcup_{u \in U}N_G(u) \setminus U$.
Similarly, $N_G[U]$ denotes $N_G(U) \cup U$.
For any vertex subset $S \subseteq V$, 
we call $G[S] = (S, E[S])$ an \name{induced subgraph},
where $E[S] = E(G) \cap (S \times S)$.
Since $G[S]$ is uniquely determined by $S$, 
we sometimes identify $G[S]$ with $S$. 
For any edge subset $E' \subseteq E$, we call $G[E'] = (V'(E'), E')$ \name{edge induced subgraph}, 
where $V'(E') = \bigcup_{\set{u, v} \in E'} u$. 
We define $G \setminus \set{e} = (V, E \setminus \set{e})$ and 
$G\setminus \set{v} = G[V\setminus\set{v}]$.
For simplicity, we use $v \in G$ and $e \in G$ to refer to $v \in V(G)$ and $e \in E(G)$, respectively. If $G$ is clear from the context, we will also use simplified notation such as $V$, $E$, $N(u)$ instead of $V(G)$, $E(G)$, $N_G(u)$.

A sequence $P = (v_1, \dots, v_{k+1})$ of distinct vertices is a \name{path} from $v_1$ to $v_{k+1}$ ($v_1$-$v_{k+1}$ path for short) in $G = (V, E)$ 
if for any $i \in [1, k]$,  $\set{v_i, v_{i+1}} \in E$. 
$P$ is a \name{shortest path} between two vertices if there is no shorter path between them. 
Let us denote by $V(P)$ and $E(P)$ the set of vertices and edges in $P$, respectively. 
We say that $G$ is \name{connected} if
for any two vertices $u, v \in V$, there is a $u$-$v$ path. 
We say that a sequence $C = (v_1,\dots, v_{k+1})$ of vertices is a \name{cycle} if 
$(v_1, \dots, v_{k})$ is a $v_1$-$v_{k}$ path, $v_{k+1} = v_1$, and $\set{v_k, v_{k+1}} \in E$. 
The \name{length} of a path or cycle is defined by its number of edges. 
The \name{distance} between two vertices is the length of a shortest path between them.  
The \name{girth} of $G$, denoted by $g(G)$, is the length of a shortest cycle in $G$. 
For simplicity, we say that $G$ has girth $k$ if $g(G)\ge k$.
The girth of acyclic graphs is usually assumed to be $\infty$.


\begin{figure}[t]
    \centering
    \includegraphics[width=0.6\textwidth]{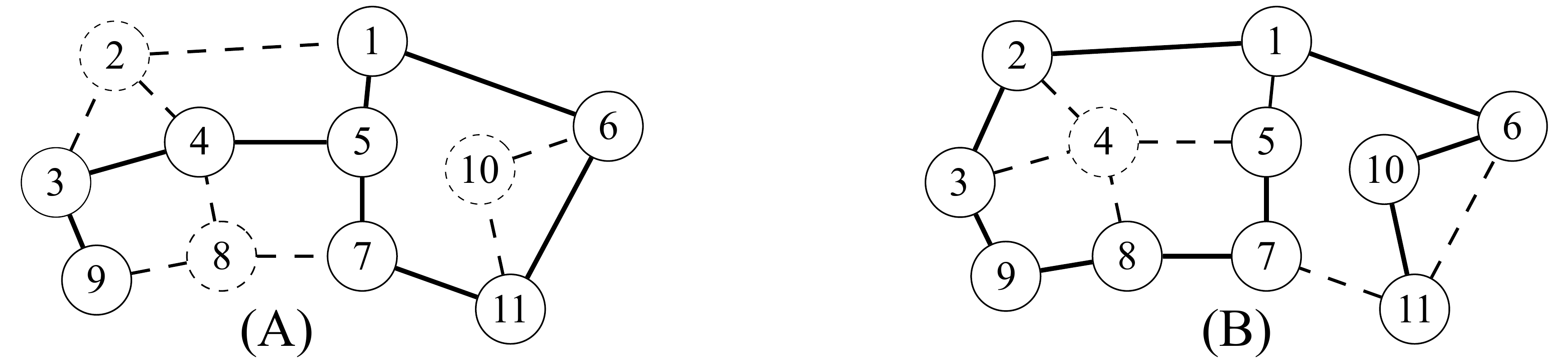}
    \caption{Dashed edges and vertices are not included by an induced subgraph and a subgraph. 
    An induced subgraph of girth five (A) and a subgraph of girth six (B). }
    \label{fig:sol}
\end{figure}

We define our problems as follows and
Fig.~\ref{fig:sol} shows examples of solutions Problem~\ref{prob:enum:ind} and Problem~\ref{prob:enum:sub}. 
If we store all outputs, then it is easy to avoid duplicates.  
Our algorithms achieve without duplicates in polynomial space. 

\begin{problem}[$k$-girth connected induced subgraph enumeration]
  \label{prob:enum:ind}
  Enumerate all connected induced subgraphs
  $S$ of a graph $G$ with $g(S)\ge k$, without duplicates.
\end{problem}

\begin{problem}[$k$-girth connected subgraph enumeration]
  \label{prob:enum:sub}
  Enumerate all connected subgraphs
  $S$ of a graph $G$ with $g(S)\ge k$, without duplicates.
\end{problem}

\section{Enumeration by binary partition}
\label{sec:basic}
The \name{binary partition method} is one of the fundamental frameworks for designing enumeration algorithms. 
Typically, a binary partition algorithm $\sig A$ has the following structure:
first $\sig A$ picks an element $x$ of the input, then 
divides the search space into two disjoint spaces, one containing the solutions that include $x$, and one those that do not. $\sig A$ recursively executes the above step until all elements are picked. 
Whenever the search space contains exactly one solution, $\sig A$ outputs it. 
We call each dividing step an \name{iteration}. 

\newcommand{\done}{\textsc{done}\xspace}

\begin{algorithm}[t]
  \caption{Enumerate all connected induced subgraphs with girth $k$.}
  \label{algo:naive}
  \Procedure(\tcp*[f]{$G$: an input graph, $k$: positive integer}){\ELG{$G, k$}}{
    \RecELG{$\emptyset,G$}\label{algo:naive:step:initial:call}\;
  }
  \Procedure(\tcp*[f]{$S$: the current solution}){\RecELG{$S, G$}}{ 
    Output $S$\;
    $\done \gets \emptyset$\;
    
    \For{$v \in \cand{S}$\label{step:for}}{
      \RecELG{$S \cup \set{v}, G\setminus \done$}\;
      \label{step:recursive:call}
      $\done \gets \done \cup \set{v}$\;\label{step:const2}
    }
    \Return \;
  }
\end{algorithm}

Algorithm \ELG, detailed in Algorithm~\ref{algo:naive}, represents a basic strategy for Problem~\ref{prob:enum:ind}.
Algorithm~\ref{algo:naive} is based on binary partition, although each iteration divides the search space in more than two subspaces.
While \ELG enumerates solutions by picking vertices on each iteration, we can obtain an enumeration algorithm for Problem~\ref{prob:enum:sub} by modifying \ELG so that it picks edges instead. 

Let $G$, $X$, and $S(X)$ be respectively
an input graph, an iteration, and the solution received by the iteration $X$.
A vertex $v \notin S(X)$ is a \name{candidate vertex} for $S(X)$ if  $g(S(X) \cup \set{v}) \ge k$ and $S(X)\cup\set{v}$ is connected, 
that is, the addition of a candidate vertex generates a new solution. 
Let $\cand{S(X)}$ be a set of candidate vertices for $S(X)$.  
We call $\cand{S(X)}$ the \name{candidate set} of $S(X)$. 
Now, suppose that 
$X$ generates new iterations $Y_1, \dots, Y_d$ 
by adding vertices in $\cand{S(X)} = \set{v_1, \dots, v_d}$ on line~\ref{step:recursive:call}. 
For each $i$, we say that $X$ is the \name{parent} of $Y_i$, 
and $Y_i$ is a \name{child} of $X$. 
Note that, on iteration $Y_i$ and its descendant iterations, 
\ELG outputs solutions that do not include $v_1, \dots, v_{i-1}$ but do include $v_i$. This implies that the solution space of $Y_i$ is disjoint from those of each $Y_{j<i}$ created so far, i.e., \ELG divides the solution space of $X$ in $d$ disjoint subspaces.
The only iteration without a parent is the one generated on line~\ref{algo:naive:step:initial:call}, which we call the \name{initial iteration} and denote by $I$. 
We remark that $S(I) = \emptyset$ and that $\emptyset$ is a solution. 

By using the above parent-child relation, 
we introduce the \name{enumeration tree} $\sig T(G) = \sig T = (\sig V, \sig E)$. 
Here, $\sig V$ is the set of iterations of \ELG for $G$ and  
$\sig E$ is a subset of $\sig V \times \sig V$. 
For any pair of iterations $X$ and $Y$, 
$(X, Y) \in \sig E$ if and only if $X$ is the parent of $Y$. 
We can observe that 
$\sig T$ has no cycles  
since every child iteration of $X$ receives a solution whose size is larger than $S(X)$. 
In addition, 
each iteration other than the initial iteration has exactly one parent. 
This implies that 
the initial iteration is an ancestor of all iterations 
and thus $\sig T$ is connected. 
Thus, $\sig T$ forms a tree. 
Next three lemmas show the correctness of \ELG. 

\begin{lemma}
  \label{lem:X}
  Let $G$ be a simple undirected graph and $k$ a positive integer. 
  Then, every output of \ELG induces a connected subgraph of
  girth $k$. 
\end{lemma}

\begin{proof}
  Let us prove the statement by induction. 
  For any iteration $X$ with $\size{S(X)} \le i$, 
  suppose that $G[S(X)]$ is connected and $g(G[S(X)]) \ge k$. 
  From the definition of $\cand{S(X)}$, $S(X)\cup\set{v'}$ is connected for any vertex $v' \in \cand{S(X)}$, and 
  $g(G[S(X) \cup \set{v}]) \ge k$ holds, thus the condition holds for all $Y$ with $\size{S(Y)}\le i+1$.
  
  Since for the initial iteration $I$, with $\size{S(I)}=0$, $G[S(I)]$ is connected\footnote{As a degenerate case, we can consider $G[\emptyset]$ connected as it contains less than two vertices.} and $g(G[S(X)]) = \infty$, the statement holds for all $i$.
\end{proof}

\begin{lemma}
  \label{lem:no:duplication}
  If $X$ and $Y$ are two distinct iterations on \ELG, then $S(X) \neq S(Y)$. 
\end{lemma}

\begin{proof}
  If $X$ is an ancestor of $Y$ in $\sig T$, then $S(Y)$ must contain $S(X) \cup \set{v}$ for some $v\in \cand{S(X)}$, thus the statement holds in this case, and it holds for the same reason if $Y$ is an ancestor of $X$.
  Otherwise, let $Z$ be the iteration on $\sig T$
  that is the lowest common ancestor of $X$ and $Y$. 
  Since $\sig T$ is a tree, $Z$ always exists. 
  Let $X'$ and $Y'$ be children of $Z$ such that they are ancestors of $X$ and $Y$ respectively. 
  Without loss of generality, 
  we can assume that $Y'$ is called after $X'$.  
  In line~\ref{step:recursive:call}, 
  when $Z$ picks a vertex $x$ from $\cand{S(Z)}$ to call $X'$, 
  $x$ is added to \done, and then $Y'$ is called on $G\setminus \done$. 
  This implies $Y'$ and its descendants can not include $x$, thus the statement holds in this case too.
\end{proof}

\begin{lemma}
  \label{lem:enum}
  Let $G$ be a simple undirected graph and $k$ a positive integer. 
  \ELG{$G,k$} outputs all connected induced subgraphs with girth $k$ in $G$ exactly once. 
\end{lemma}

\begin{proof}
    By Lemma~\ref{lem:X}, 
    \ELG outputs only solutions, and by Lemma~\ref{lem:no:duplication} 
    it does not output each solution more than once. 
    We show that \ELG outputs all solutions by induction. 
    Let $S$ be a solution.
    If $\size{S} = 0$, \ELG outputs the empty set. 
    
    Otherwise, there is an iteration $X_0$ such that $S(X_0)\subseteq S$ and $S\subseteq V(G)$ (that is, no vertex of $S$ has been removed from $G$). This is trivially true, e.g. for $X_0 = I$, since $S(I) = \emptyset$ and nothing has been removed from $G$.
    Note that every subgraph of a graph with girth at least $k$ must also have girth at least $k$, thus every $v\in S\setminus S(X_0)$ such that $G[S(X_0)\cup\set{v}]$ is connected must be in $\cand{S(X_0)}$. As $S$ is connected there is at least one such $v$ in $\cand{S(X_0)}$.
    
    Consider the first execution of Line~\ref{step:recursive:call} in $X$ for which a vertex $v\in S\setminus S(X_0)$ is considered to generate a child iteration $X_1$. As no vertex of $S$ was added to \done in $X_0$, we still have that $S(X_1)\subseteq S$ and $S\subseteq V(G)$ in iteration $X_1$, but $|S(X_1)| = |S(X_0)|+1$. Hence, by induction, \ELG will eventually find $S$. 
\end{proof}


Using Itai's algorithm~\cite{Itai:Rodeh:SIAM:1978} to compute the girth of a graph in $\order{mn}$, we can obtain a first trivial complexity bound for Algorithm~\ref{algo:naive}.

\begin{theorem}
  \label{theo:naive}
    \ELG solves Problem~\ref{prob:enum:ind} with delay $\order{n^2m}$.
\end{theorem}

\begin{proof}
  By Lemma~\ref{lem:enum}, 
  \ELG enumerates all solutions without duplication. 

  As for its delay, since every iteration outputs a solution, it is sufficient to bound the time complexity of one iteration.  
  The bottleneck of \RecELG is Line~\ref{step:for}: in order to compute $\cand{S(X)}$, \ELG must iterate over all vertices $v\in V(G)$ and check whether the girth of $G[S(X) \cup \set{v}]$ is $k$.

  By using time Itai's algorithm~\cite{Itai:Rodeh:SIAM:1978}, we can test each $v$ in $\order{nm}$, thus the total cost is bounded by $|V(G)|\cdot \order{nm} = \order{n^2m}$.
\end{proof}

\paragraph{\textbf{Non-induced, weighted, and non-connected case.}}
Let us briefly show how \ELG also applies to some variants of the problem.
Firstly, we can solve Problem~\ref{prob:enum:sub}, i.e., enumerate \textit{edge} subgraphs, by modifying \ELG as follows:
Each solution is a set of edges $S\subseteq E$, and the candidate set $\cand{S(X)}$ becomes $\cand{S(X)} = \inset{e \in E(X)}{G[S(X) \cup \set{e}] \text{ is connected and }$ $g(G[S(X) \cup \set{v}]) \ge k}$.
It is straightforward to see that Lemma~\ref{lem:enum} still holds (replacing the word \textit{induced} with \textit{edge} in the statement), and that the modified algorithm will solve Problem~\ref{prob:enum:sub} in polynomial delay and polynomial space. 

Furthermore, we can consider the \textit{weighted} version of the problem, where the length of a cycle is the sum of the weights of its edges: we can find the girth in this case by adapting the Floyd-Warshall algorithm, and thus still enumerate all solutions for both the induced and edge subgraph version of the problem, in polynomial delay and polynomial space. 

Finally, we consider non-connected case, i.e., where the solutions are all induced or edge subgraphs of girth $k$, and not just the connected ones: this is trivially done by redefining the candidate set as $\cand{S(X)} = \inset{v \in V(G)}{g(G[S(X) \cup \set{v}]) \ge k}$ for Problem~\ref{prob:enum:ind}, and similarly for Problem~\ref{prob:enum:sub}. 
If $G[S]$ is not connected, its girth is the minimum among that of its connected components, thus we can still use Itai's algorithm (or Floyd-Warshall if weighted edges are considered as well), and again obtain polynomial delay and polynomial space.


\section{Induced subgraph enumeration}
\label{sec:propose}
The bottleneck of \ELG  
is the computation of the candidate set. 
In this section, 
we present a more efficient algorithm \ELGIS for Problem~\ref{prob:enum:ind}. 
\ELGIS is based on \ELG, but each iteration exploits information from the parent iteration, and maintains distances in order to improve the computation of the candidate set. The procedure is shown in Algorithm~\ref{algo:elg}. 

\begin{algorithm}[t]
  \caption{Updating data structures in \ELGISname. }
  \label{algo:elg}
  \Procedure{\NextC{$v, \cand{S}, \DistM{S}, \TadM{S}, S, k, G$}}{
    $\cand{S \cup \set{v}}  \gets$ \UpdateCand{$v, S$}\;
    $\DistM{S \cup \set{v}} \gets$ \UpdateOne{$v, \cand{S \cup \set{v}}$}\;
    $\TadM{S \cup \set{v}}  \gets$ \UpdateTwo{$v, \cand{S \cup \set{v}}$}\;
  }
  \Fn{\UpdateCand{$v, S$}}{
    $\cand{S\cup\set{v}} \gets N(v) \cup \cand{S}$\;
    \ForEach{$u \in \cand{S}$}{
      \lIf{$\DistM[uv]{S} + \TadM[uv]{S} \ge k$}{
        $\cand{S\cup\set{v}} \gets \cand{S\cup\set{v}} \cup \set{u}$
      }
    }
    \Return $\cand{S\cup\set{v}}$\;
  }
  \Fn{\UpdateOne{$v, \cand{S\cup\set{v}}$}}{
    \ForEach{$u \in \cand{S\cup\set{v}} \cup S, w \in \cand{S\cup\set{v}}$}{ 
      $\DistM[uw]{S} \gets \min\set{\DistM[uw]{S}, \DistM[uvw]{S}}$
    }
    \Return $\DistM{S \cup\set{v}}$
  }
  \Fn{\UpdateTwo{$v, \cand{S \cup \set{v}}$}}{
    \ForEach{$u, w \in \cand{S \cup \set{v}}$}{ 
      $p_1 \gets \min\set{\DistM[uw]{S}, \DistM[uvw]{S \cup \set{v}}, \TadM[uw]{S}}$\;
      $p_2 \gets \text{the second smallest length in }\set{\DistM[uw]{S}, \DistM[uvw]{S \cup \set{v}}, \TadM[uw]{S}}$\;
      \If(\tcp*[f]{$x \in N(u)\cap S \cup \set{v}$}){$p_1 + p_2 \ge k$}{
        $p_2 \gets \text{the second smallest length in }\set{\DistM[xw]{S \cup \set{v}} + 1}$\;
      }
      $\TadM[uw]{S \cup \set{v}} \gets p_2$\;
    }
    \Return $\TadM{S \cup \set{v}}$\;
  }
\end{algorithm}


\ELGIS uses the second distance between vertices defined as follows. 
Let $v$ be a vertex in $\cand{S} \cup S$, and 
$u$ and $u'$ be vertices in $\cand{S}$. 
We denote by $\DistM[uv]{S}$ the distance between $v$ and $u$ in $G[S \cup \set{v, u}]$, 
and by $\TadM[uu']{S}$ the distance between $u$ and $u'$ in $G[S \cup \set{u, u'}] \setminus \set{e_0}$, 
where $e_0 = (u, \cdot)$ is the first edge on a shortest path between $u$ and $u'$. 
Note that for any vertices $x \in G\setminus\set{\cand{S}\cup S}$, $y\in G \setminus \cand{S}$, and $y'\in G \setminus \cand{S}$, 
$\DistM[xy]{S} = \infty$ and $\TadM[yy']{S} = \infty$. 
Especially, 
we call $\TadM[uu']{S}$ the \name{second distance} between $u$ and $u'$ in $G[S \cup \set{u, u'}]$. 
In addition, 
we call a path whose length is the second distance a \name{second shortest path}. 
Moreover, we write $\DistM[uwv]{S}$ and $\TadM[uwv]{S}$ 
for the distance and the second distance from $u$ to $v$ via a vertex $w$, respectively. 
Let $P$ and $P'$ be respectively a $v$-$u$ shortest path and a $v$-$u$ second shortest path. 
Since $P$ and $P'$ do not share $e_0$ but do share their ends,  
$H$ must have a cycle including $v$ and $u$, 
where $H$ is a subgraph of $G$ such that $V(H) = V(P)\cup V(P')$ and $E(H) = E(P) \cup E(P')$. 
Fig.~\ref{fig:dist} (C) shows an example of a cycle made by $P$ and $P'$. 
To compute the candidate set efficiently, 
we will use the following lemmas. 
In the following lemmas, 
let $X$ and $Y$ be two iterations such that $X$ is the parent of $Y$, 
and $v$ be a vertex in $\cand{S(X)}$ such that $S(Y) = S(X) \cup \set{v}$. 

\begin{figure}[t]
    \centering
    \includegraphics[width=0.7\textwidth]{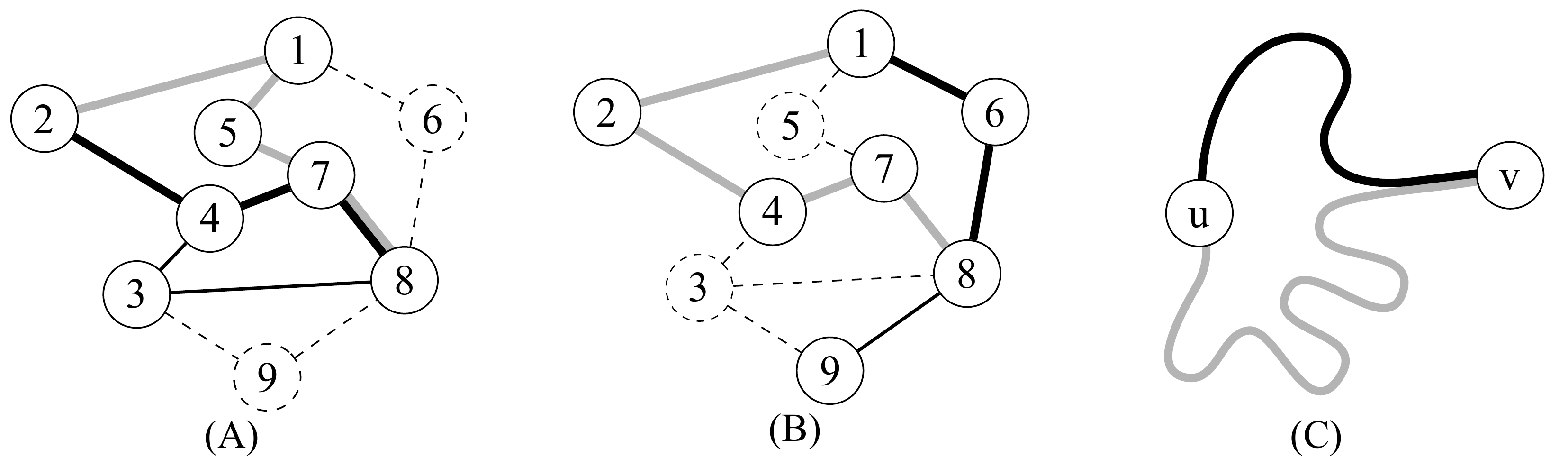}
    \caption{
    (A) and (B) show two induced subgraphs. 
    (C) shows a shortest path and a second shortest path.  
    Dashed edges and vertices are not contained by induced subgraphs. 
    Black and gray paths show respectively shortest and second shortest paths. 
    }
    \label{fig:dist}
\end{figure}

\begin{lemma}
  \label{lem:dist}
  Let $u$ and $w$ be two vertices in $\cand{S(X)}$ and 
  $k= g(G[S(X)])$. 
  (A) $g(G[S(X) \cup \set{u, w}]) \ge k$ 
  if and only if
  (B) $\DistM[uw]{S(X)} + \TadM[uw]{S(X)} \ge k$. 
\end{lemma}

\begin{proof}
  Clearly, (A) $\to$ (B) holds by definition of $\DistM{S(X)}$ and $\TadM{S(X)}$. 
  For the direction (B) $\to$ (A), consider a shortest cycle $C$ in $G[S(X) \cup \set{u, w}])$ in the following three cases: 
  (I) $u, w \notin C$: 
  $\size{C} \ge k$ since $g(G[S(X)])\ge k$.
  (II) Either $u$ or $w$ in $C$: 
  $\size{C} \ge k$ since $u$ and $w$ belong to $\cand{S(X)}$. 
  (III) Both $u$ and $w$ in $C$: 
  $C$ can be decomposed into two $u$-$w$ paths $P$ and $Q$. 
  Without loss of generality, $\size{P} \le \size{Q}$. 
  If $P$ is a $u$-$w$ shortest path, then $\size{C} \ge k$ from (B),
  since $Q$ is at least as long as the \textit{second distance} $\TadM[uw]{S(X)}$. 
  Otherwise, there is a $u$-$w$ shortest path $P'$ and 
  a cycle $C'$ consisting of a part of $P$ (or $Q$) and a part of $P'$. 
  If $C'$ contains $w$, then $\size{C'} = \size{C} \ge k$ since $C$ is a shortest cycle. 
  If $C'$ does not contain $w$, then $\size{C'}$ is a cycle in $G[S(X) \cup \set{u}]$, thus $\size{C'}\ge k$ because $u\in \cand{S(X)}$.
\end{proof}


\begin{lemma}
  \label{lem:candX1}
  \ELGIS computes $\cand{S(Y)}$ in 
  $\order{\size{\cand{S(X)}} + \size{N(v)}}$ time. 
\end{lemma}

\begin{proof}
  From Lemma~\ref{lem:dist}, 
  vertex $u$ in $\cand{S(X)}$ belongs to $\cand{S(Y)}$ if and only if 
  $\DistM[uv]{S(X)} + \TadM[uv]{S(X)} \ge k$. 
  This can be done in constant time. 
  In addition, from the connectivity of $G[S(Y)]$, 
  $\cand{S(Y)} \setminus \cand{S(X)} \subseteq N(v)$. 
  Thus, we can find $\cand{S(Y)} \setminus \cand{S(X)}$ in 
  $\order{\size{\cand{S(X)}} + \size{N(v)}}$ time. 
\end{proof}

Next, we consider how to update the values of $\DistM{S(Y)}$ and $\TadM{S(Y)}$  when adding $v$ to $S(X)$. 
We can update the old distances to the ones after adding $v$ as in the Floyd-Warshall algorithm (see Algorithm~\ref{algo:elg}), meaning that we can compute $\DistM{S(Y)}$ in $\order{\size{S(X)\cup\cand{S(X)}} \cdot \size{\cand{S(X)}}}$ time. 
By the following lemma, the values of $\TadM{S(Y)}$
can be updated in $\order{\size{S(Y)}}$ time for each pair of vertices in $\cand{S(Y)}$.

\begin{lemma}
  \label{lem:update1}
  Let $u$ and $w$ be two vertices in $\cand{S(X)}$,  
  $e_0$ be an edge in a $u$-$w$ shortest path in $G[S(X) \cup \set{u, w}]$, 
  and $H = G[S(X) \cup \set{u, w}] \setminus \set{e_0}$. 
  If $N_H(u) = \emptyset$, then $\TadM[uw]{S(X)} = \infty$. 
  Otherwise, $\TadM[uw]{S(X)} = \min_{y \in N_H(u)}\set{\DistM[yw]{S(X)} + 1}$. 
\end{lemma}

\begin{proof}
    From the definition of $\TadM[uw]{S(X)}$, 
    if $N_H(u) = \emptyset$, 
    then $\TadM[uw]{S(X)} = \infty$. 
    We assume $\size{N_H(u)} \ge 1$. 
    Since $u \notin S(X)$, every shortest path between $u$ and $w$ in $G[S(X) \cup \set{w}] \cup {f}$ 
    contains $f$, where $f = \set{u, y}$. 
    Hence, 
    $\DistM[yw]{S(X)} + 1$ is equal to the distance between $u$ and $w$ 
    in $G[S(X) \cup \set{w}] \cup \set{f}$. 
    Hence, the statement holds. 
\end{proof}
The next lemma implies that 
if $\DistM[uw]{S(X)} + \TadM[uw]{S(X)} < k$, 
i.e., $G[S(X) \cup \set{u, w}]$ is not a solution, 
then computing $\TadM[uw]{S(Y)}$ takes constant time.

\begin{figure}[t]
    \centering
    \includegraphics[width=0.7\textwidth]{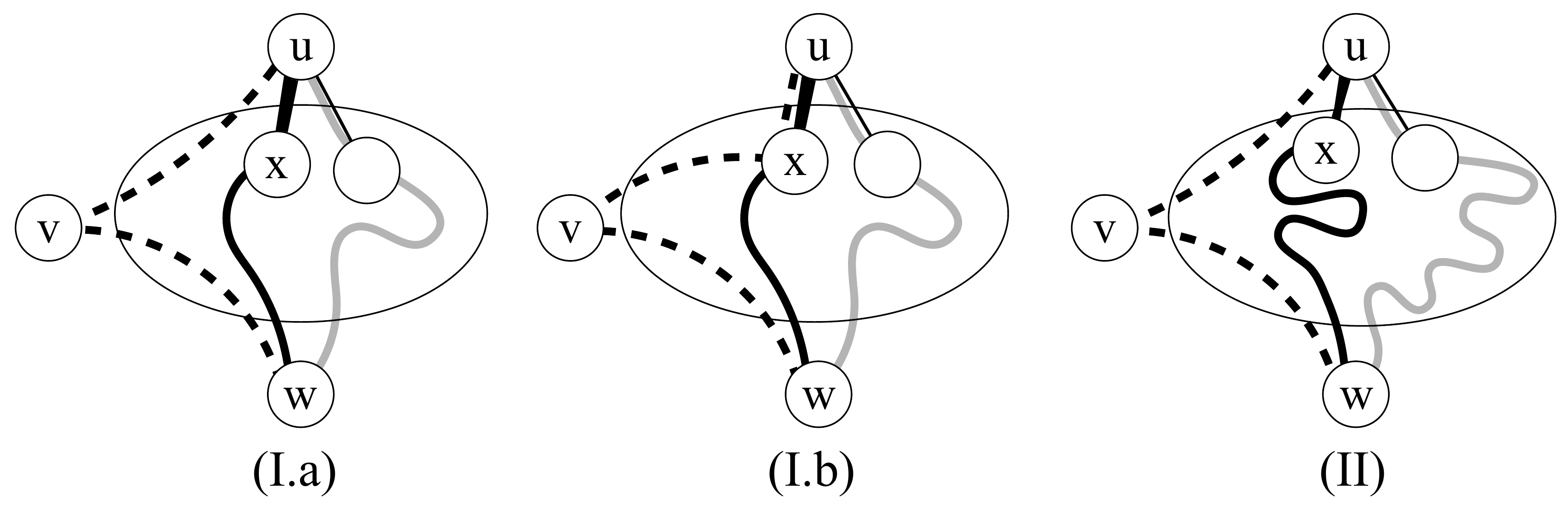}
    \caption{Examples of each case in Lemma~\ref{lem:update2}. 
    Solid lines are $u$-$v$ shortest paths in $G[S(X) \cup \set{u, w}]$. 
    Gray solid lines are $u$-$v$ second shortest paths in $G[S(X) \cup \set{u, w}]$. 
    Dashed lines are $u$-$v$-$w$ shortest paths in $G[S(Y) \cup \set{u, w}]$. 
    Let $\set{u, x}$ be the first edge in a shortest path: the sum of lengths of a solid and gray solid line is less than $k$. 
    }
    \label{fig:all_case}
\end{figure}

\begin{lemma}
    \label{lem:update2}  
    Let $u$ and $w$ be two vertices in $\cand{S(Y)}$. 
    If $p_1 + p_3 < k$, 
    then $\TadM[uw]{S(Y)} = \min\set{\max\set{p_1, p_2}, p_3}$, 
    where $p_1 = \DistM[uw]{S(X)}$, 
    $p_2 = \DistM[uvw]{S(Y)}$, 
    and $p_3 = \TadM[uw]{S(X)}$. 
\end{lemma}

\begin{proof}
    Let $G_X = G[S(X) \cup \set{u, w}]$ and $G_Y = G[S(Y) \cup \set{u, w}]$.  
    Note that $p_1 \le p_3$. 
    We consider the following cases: 
    (I) $p_1 < p_2$:
    Let $e = \set{u, x}$ be the first edge of a $u$-$w$ shortest path $P$ in $G_Y$. 
    Note that $P$ cannot contain $v$. 
    (I.a) There exists a $u$-$v$-$w$ shortest path $Q$ that does not contain $e$:  
    clearly, $\TadM[uw]{S(Y)} = \min\set{\size{Q} = p_2, p_3}$.  
    (I.b) Every $u$-$v$-$w$ shortest path $Q$ contains $e$:  
    there always exists a cycle $C$ in $S(Y) \cup\set{w}$ such that $V(C) \subseteq (V(P) \cup V(Q))\setminus \set{u}$ and $C$ does not contain $u$. 
    Note that $\size{C} < p_1 + p_2$. 
    If $p_2 \le p_3$, 
    then this contradicts $w \in \cand{S(Y)}$ since $\size{C} < k$. 
    Thus, $p_2 > p_3$. 
    This implies that $\size{Q} - 1 \ge p_3$. 
    Hence, $\TadM[uw]{S(Y)} = p_3$. 
    (II) $p_2 \le p_1$:
    this assumption implies that there exists a $u$-$w$ shortest path $P$ in $G_Y$ that contains $v$, and $p_1 + p_2 < k$. 
    Let $e$ be the first edge of $P$ in $G_Y$  
    and $Q$ be a $u$-$v$-$w$ shortest path in $G_Y \setminus \set{e}$. 
    Now, we can see 
    $\size{Q} > p_1$ since if $\size{Q} \le p_1$, 
    then $u \notin \cand{S(Y)}$ since $P$ and $Q$ make a cycle $C$ containing $u$ with $\size{C} < k$. 
    Thus, the length of a $u$-$w$ shortest path in $G_Y\setminus\set{e}$ is $p_1$, 
    and $\TadM[uw]{S(Y)} = p_1$ holds. 
\end{proof}



Algorithm~\ref{algo:elg} shows in detail the update of the candidate set, 
$\DistM{\cdot}$, and $\TadM{\cdot}$ (done using Lemma~\ref{lem:update2}).
%
We analyze the time complexity of \ELGIS. 
Let $ch(X)$ be the set of children of $X$ and
$\#gch(X)$ be the number of grandchildren of $X$. 
The next lemma shows the time complexity for updating $\TadM{S(X)}$. 

\begin{lemma}
  \label{lem:comp}
  We can compute $\TadM{S(Y)}$ from $\TadM{S(X)}$ in 
  $\order{\#gch(Y)\cdot\size{S(Y)} + \size{\cand{S(Y)}}^2}$ time. 
\end{lemma}

\begin{proof}
  Let $u$ and $w$ be two vertices in  $\cand{S(Y)}$. 
  Two cases are possible:

  \noindent(I)  $\DistM[uw]{S(X)} + \TadM[uw]{S(X)} \ge k$:  
  By Lemma~\ref{lem:update1}, 
  computing $\TadM[uw]{S(Y)}$ takes $\order{\size{S(Y)}}$ time, checking only vertices in $S(Y)$. 
  As the number of pairs $(u, w)$ that fit this case is bounded by $\#gch(Y)$, \ELGIS needs 
  $\order{\#gch(Y)\cdot\size{S(Y)}}$ time to compute this part.
  (II) $\DistM[uw]{S(X)} + \TadM[uw]{S(X)} < k$:  
  From Lemma~\ref{lem:update2}, 
  computing $\TadM[uw]{S(Y)}$ takes constant time, for a total complexity of
  $\order{\size{\cand{S(Y)}}^2}$, which proves the statement. 
\end{proof}

\begin{theorem}
  \label{theo:fast}
  \ELGIS enumerates all solutions in
  $\order{\sum_{S \in \sig S}\size{N[S]}}$ time using
  $\order{\max_{S \in \sig S}\set{\size{N[S]}^3}}$ space, where $\sig S$ is the set of all solutions.
\end{theorem}

\begin{proof}
  The correctness of \ELGIS follows from Lemma~\ref{lem:enum}. 
  We first consider the space complexity. 
  In an iteration $X$, 
  \ELGIS uses $\order{\size{\cand{S(X)} \cup S(X)}^2}$ 
  space for storing values of $\DistM{\cdot}$ and $\TadM{\cdot}$. 
  In addition, the height of $\sig T$ is at most $\max_{S \in \sig S}\set{\size{S}}$.
  Therefore, \ELGIS uses $\order{\max_{S \in \sig S}\set{\size{N[S]}^3}}$ space.

  Let 
  $c(X)$ be $\size{\cand{S(X)}}$ and  
  $T(X, Y)$ be the time needed to generate $Y$ from $X$, i.e., an execution of $\NextC()$ (Algorithm~\ref{algo:elg}). 
  From Lemma~\ref{lem:candX1}, Lemma~\ref{lem:update1}, and the Floyd-Warshall algorithm, 
  $T(X, Y)$ is 
  $\order{c(X) + \size{N(v)} + c(Y)\cdot\size{S(X)} + \#gch(Y)\cdot\size{S(Y)} + c(Y)^2}$ time. 
  In addition, 
  $\size{N[S(X)]} \le \size{N[S(Y)]}$, 
  $\size{N(v)} = \order{\size{N[S(Y)]}}$, and 
  $c(X) = \order{N[S(X)]}$ since every vertex in the candidate set 
  has a neighbor in $S(X)$. 
  Thus, 
  $T(X, Y) = \order{\size{N[S(Y)]}(c(Y) + \#gch(Y))}$ time. 
  Note that the sum of children and grandchildren for all iterations is at most $2\size{\sig V}$. 
  Thus, 
  by distributing the $\order{\size{N[S(Y)]}}$ time from $X$ to children and grandchildren of $Y$, 
  each iteration needs $\order{\size{N[S(Y)]}}$ time since each iteration receives costs only from the parent and the grandparent. 
  In addition, 
  each iteration outputs a solution, and hence
  the total time is 
  $\order{ \sum_{S \in \sig S}\size{N[S]}}$. 
\end{proof}

\section{Subgraph enumeration}
We propose an algorithm, \ELGS, for enumerating all subgraphs with girth $k$ in a given graph $G$, detailed in Algorithm~\ref{algo:elgs}.
A trivial adaptation of \ELGIS would run in $\order{m}$ time per solution,
as the candidate sets are sets of edges, whose size is $\order{m}$.
To improve this running time, 
\ELGS selects candidates in a certain order, so that the number of candidate edges does not exceed no more than the number of nodes in the previous solution $G[S]$. 

\begin{figure}[t]
    \centering
    \includegraphics[width=0.7\textwidth]{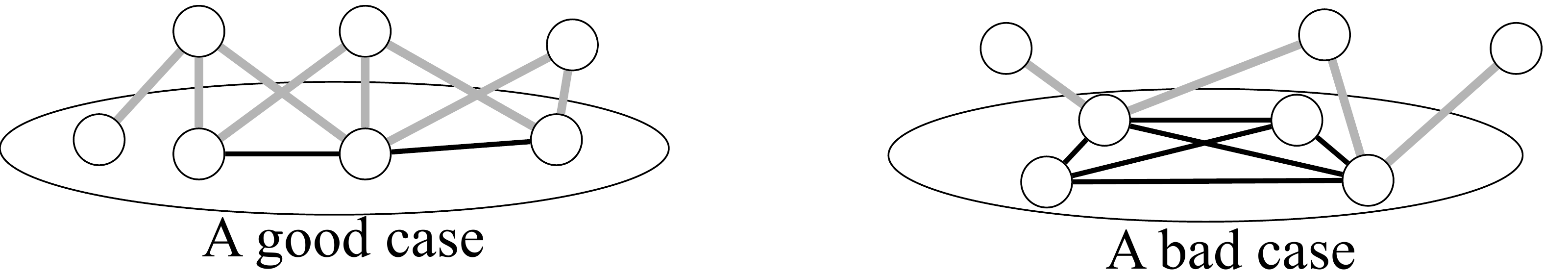}
    \caption{Black solid lines and gray solid lines represent inner edges and outer edges, respectively. %
    Our main strategy is to reduce the number of inner edges in \texttt{EBG-S}. }
    \label{fig:inner}
\end{figure}

Let $S$ be the current solution. 
Note that $S$ is an edge set. 
We first define an inner edge and an outer edge as follows: 
an edge $e = \set{u, v}$ is an \name{inner edge} for $S$ if $u, v \in G[S]$, 
and an \name{outer edge} otherwise (see Fig.~\ref{fig:inner}).
Let $\cin{S}$ and $\cout{S}$ be a set of inner edges and outer edges in $\cand{S}$, respectively. 
We first consider the case when \ELGS picks an outer edge. 
In the following lemmas, 
let $X$ be an iteration in enumeration tree $\sig T$, 
$e$ be an edge not in $X$, and 
$Y$ be the child iteration of $X$ satisfying $S(Y) = S(X) \cup \set{e}$.

\begin{lemma}
  \label{lem:out}
  Let $e = \set{x, y}$ be an outer edge such that $x \in V(G[S(X)])$. 
  Then $\cand{S(Y)} \subseteq (\cand{S(X)} \cup E(y)) \setminus \set{e}$, where $E(y)$ are the edges incident to $y$. 
\end{lemma}


\begin{proof}
    An edge $g \notin E(y) \cup \cand{S(X)}$ may not be added to $S(Y)$ as the resulting subgraph would be disconnected, and $e\not\in\cand{S(Y)}$ since $e\in S(Y)$. 
\end{proof}

From Lemma~\ref{lem:out},
\ELGS manages the candidate set $\cand{S(Y)}$ 
in $\order{\size{\cand{S(Y)}} + \size{V(G[S(X)])}}$ time
when \ELGS picks an outer edge $e$ 
since we can add all edges 
$e' \notin S(X) \cup \cand{S(X)}$ incident to $y$ 
and $S(Y) \cup \set{e'}$ is a solution. 
Moreover, removed edges are at most $\size{V(G[S(X)])}$ since 
all removed edges have a vertex in $V(G[S(X)])$. 
In this case, 
\ELGS can obtain $\cin{S(Y)}$ and $\cout{S(Y)}$ 
in $\order{S(X)}$ time and $\order{\cand{S(Y)}}$ time, respectively. 
Next, we consider that when \ELGS picks an inner edge $e$. 
When we pick an inner edge, $\cand{S(Y)}$ is monotonically decreasing.

\begin{lemma}
\label{lem:in}
  If $e$ is an inner edge, 
  then $\cin{S(Y)} \subset \cin{S(X)}$ and 
  $\cout{S(Y)} = \cout{S(X)}$.
\end{lemma}

\begin{proof}
    Since $e$ is an inner edge $V(G[S(Y)]) = V(G[S(X)])$, 
    thus
    there is no edge $f \in \cin{S(Y)} \setminus \cin{S(X)}$. 
    Since $e \notin \cin{S(Y)}$ and 
    no edge in $\cout{S(X)}$ is in $\cin{S(Y)}$,
    $\cin{S(Y)} \subset \cin{S(X)}$.
    Moreover, there is no cycle including $f \in \cout{S(X)}$ 
    in $G[S(Y) \cup \set{f}]$, 
    hence $\cout{S(Y)} = \cout{S(X)}$. 
\end{proof}

\begin{algorithm}[t]
  \caption{Updating data structures in \ELGSname. }
  \label{algo:elgs}
    \Procedure{\NextC{$\cand{S}, \DistMS{S}, S, k, G$}}{
        \lIf{$\cin{S} \neq \emptyset$}{ 
        $e \gets \cin{S}$; 
        \textbf{else} 
        $e \gets \cout{S}$
        }
        $\cand{S \cup \set{e}} \gets$ \UpdateCand{$e, S$}\;
        $\DistMS{S \cup\set{e}} \gets$ \UpdateThree{$e, \cand{S\cup\set{e}}$}\;
    }
    \Fn{\UpdateCand{$e = \set{u, v}, S$}}{
        \uIf{$e \in \cin{S}$}{       
            \For{$f \in \cin{S} \setminus \set{e}$}{
                \lIf{$g(G[S \cup \set{e, f}]) \ge k$}{
                    $\cin{S} \gets \cin{S} \cup \set{f}$
                }
            }
        }\Else(\tcp*[f]{We assume $u \in G[S]$ and $v \notin G[S]$}){
            \For(\tcp*[f]{Let $f$ be an edge $\set{v, w}$}){$w \in N(v)$}{
                \lIf{$g(G[S \cup \set{e, f}]) < k$}{
                    $\cout{S} \gets \cout{S} \setminus f$
                }\ElseIf{$w \in G[S]$}{
                    $(\cin{S}, \cout{S}) \gets (\cin{S} \cup f, \cout{S} \setminus f)$
                }\lElse{
                    $\cout{S} \gets \cout{S} \cup f$
                }
            }
        }
        \Return $\cin{S} \cup \cout{S}$\;
    }
    \Fn{\UpdateThree{$e = \set{u, v}, \cand{S\cup\set{e}}$}}{
    $A = \inset{v \in V(G[S])}{v \text{ is incident to } \cand{S}.}$\;
      \For(\tcp*[f]{If $e \in \cout{S}$, then $u \in V(G[S]), v \notin V(G[S])$}){$x, y \in A$}{ 
            
            \If{$e \in \cin{S}$}{
                $\DistMS[xy]{S} \gets \min\set{\DistMS[xy]{S}, \DistMS[xu]{S} + \DistMS[vy]{S} + 1, \DistMS[xv]{S} + \DistMS[uy]{S} + 1}$\;
            }\lElse{
                $\DistMS[xy]{S} \gets \min\set{\DistMS[xy]{S}, \DistMS[xu]{S} + 1}$
            }
        }
    \Return $\DistMS{S}$\;    
    }
\end{algorithm}

Next, 
    for any pair of edges $e$ and $f$ not in $G[S(X)]$, 
we consider the computation of the girth of $G[S(X) \cup \set{e, f}]$ in \ELGS. 
Let $A(X) = \inset{v \in V(G[S(X)])}{E(v) \cap \cand{S(X)} \neq \emptyset}$. 
In a similar fashion as \ELGIS, 
\ELGS uses $\DistMS{S(X)}$ for $A(X)$. 
The definition of $\DistMS{S(X)}$ is as follows: 
For any pair of vertices $u$ and $v$ in $A(X)$,
$\DistMS[uv]{S(X)}$ is the distance between $u$ and $v$ in $A(X)$. 
Note that a shortest path between $u$ and $v$ may contain a vertex in $G[S] \setminus A(X)$. 
The next lemma shows that 
by using $\DistMS{S(X)}$, 
we can compute
$\cand{S(Y)}$ in $\order{\size{V(G[S(Y)])}}$ time from $\cand{S(X)}$. 

\begin{lemma}
\label{lem:delta}
    For any iteration $X$, 
    $\size{\cin{S(X)}} \le \size{V(G[S(X)])}$. 
\end{lemma}

\begin{proof}
    The proof follows from these facts:
    (A) Initially, $\cin{S(X)}=\emptyset$. (B) Choosing $e \in \cin{S(X)}$ decreases $|\cin{S(Y)}|$. (C) $e = \{x,y\} \in \cout{S(X)}$ is chosen iff $\size{\cin{S(X)}}=0$, and (assuming wlog $y\not\in V(G[S(X)])$) it increases $|\cin{S(Y)}|$ by at most $\size{\set{ \{y,z\} : z\in V(G[S(X)])}} < \size{V(G[S(X)])}$.
\end{proof}

\begin{lemma}
    \label{lem:move_cout}
    $\size{\cout{S(X)} \setminus \cout{S(Y)}} + \size{\cout{S(Y)} \setminus \cout{S(X)}} \le \size{V(G[S(Y)])}$. 
\end{lemma}

\begin{proof}
    We consider two cases: 
    (I) $\cin{S(X)} \neq \emptyset$:
    \ELGS picks $e \in \cin{S(X)}$, and thus, 
    From Lemma~\ref{lem:in}, $\cout{S(Y)} = \cout{S(X)}$. 
    (II) $\cin{S(X)} = \emptyset$: 
    \ELGS picks $e = \set{u, v} \in \cout{S(X)}$. 
    Without loss of generality, we can assume that $u \in V(G[S(X)])$ and $v \notin V(G[S(X)])$. 
    Let $f$ be an edge $\set{v, w}$ incident to $v$. 
    Now, $w \in V(G[S(Y)])$. 
    This implies that the number of edges that are added to $\cout{S(Y)}$ and removed from $\cout{S(X)}$ 
    is at most $\size{V(G[S(Y)])}$.
\end{proof}

Note that $\size{V(G[S(X)])} \le \size{V(G[S(Y)])}$. 
Hence, from the above lemmas, 
we can obtain the following lemma. 

\begin{lemma}
    \label{lem:comp:cand}
    $\cand{S(Y)}$ can be computed in 
    $\order{\size{V(G[S(Y)])}}$ time from $\cand{S(X)}$. 
\end{lemma}

\begin{theorem}\label{theo:elgs}
\ELGS enumerates all connected subgraphs with girth $k$ in \linebreak $\order{\sum_{S \in \sig S}\size{V(G[S])}}$ total time using $\order{\max_{S \in \sig S}\set{\size{V(G[S])}^3}}$ space. 
\end{theorem}
\begin{proof}
    The proof can be obtained by adapting that of Theorem~\ref{theo:fast}. A more detailed proof can be found in the appendix. 
\end{proof}


\begin{proof}
    From Lemma~\ref{lem:enum}, the correctness of \ELGS holds. 
    Let $\sig T = (\sig V, \sig E)$ be the enumeration tree made by \ELGS. 
    We first consider the space complexity of \ELGS. 
    In each iteration $X$, 
    \ELGS needs $\order{\max_{X \in \sig V}\set{\size{A(X)}^2}}$ for storing $\DistMS{S(X)}$. 
    In addition, 
    the height of $\sig T$ is 
    $\order{\max_{S \in \sig S}\set{\size{V(G[S])}}}$. 
    \ELGS traverses on $\sig T$ in a DFS manner. 
    Hence, the space complexity of \ELGS is 
    $\order{\max_{S \in \sig S}\set{\size{V(G[S])}}^3}$. 
    
    We next consider the time complexity of \ELGS. 
    Suppose that we add $e = \set{u, v}$ to $S(X)$, 
    and $S(Y) = S(X) \cup \set{e}$, that is, $Y$ is a child iteration of $X$. 
    Then, 
    $\DistMS[xy]{S(Y)} = 
    \min\set{\DistMS[xy]{S(X)}, \DistMS[xu]{S(X)} + \DistMS[vy]{S(X)} + 1, \DistMS[xv]{S(X)} + \DistMS[uy]{S(X)} + 1}$. 
    Thus, we can compute $\DistMS{S(Y)}$ from $\DistMS{S(X)}$ 
    in $\order{\size{A(Y)}^2}$ time since 
    each value of $\DistMS{S(Y)}$ can be computed in constant time. 
    From Lemma~\ref{lem:comp:cand}, 
    \ELGS needs 
    $\order{\size{V(G[S(X)])}  + \size{A(Y)}^2}$
    time for generating data structures for $S(Y)$ from those for $S(X)$. 
    Thus, 
    since $\size{V(G[S(X)]} \le \size{V(G[S(Y)])}$,  
    the total time of \ELGS is $\order{\sum_{X \in \sig V}\size{V(G[S(X)])} + \size{A(X))}^2}$. 
    Note that, $X$ has $\size{\cand{S(X)}}$ child iterations. 
    Moreover, $\size{A(X)}$ is at most $2\size{\cand{S(X)}}$ since 
    each vertex in $A(X)$ is incident to at least one edge in $\cand{S(X)}$. 
    Hence, $\order{\size{A(X)}^2} = \order{\size{\cand{S(X)}} \size{A(X)}}$.  
    Since $\size{\sig V} = 1 + \sum_{X \in \sig V}{\size{\cand{S(X)}}}$, 
    by delivering $\order{A(X)}$ time to each child of $X$, 
    the time complexity of \ELGS is 
    $\order{\sum_{X \in \sig V}(\size{V(G[S(X)])} +  \size{A(X)})}$. 
    In addition, 
    $\size{A(X)}$ is at most $\size{V(G[S(X)])}$ since  $A(X) \subseteq V(G[S(X)])$. 
    Hence, the statement holds. 
\end{proof}


\section{Conclusion}
In this paper, we addressed the
$k$-girth connected induced/edge subgraph enumeration problems.
We proposed two algorithms: \ELGIS for induced subgraphs and \ELGS for edge subgraphs. Both algorithms have $\order{n}$ time delay and require $\order{n^3}$ space (exact bounds are reported in Table~\ref{tab:summary}).
%
The algorithms can easily be adapted to relax the connectivity constraint and consider weighted graphs.
Other possibilities include applying the algorithms for network analysis and considering the more challenging problem of enumerating 
maximal subgraphs. 

\begin{table}[t]
    \centering
    \renewcommand{\arraystretch}{1.3}
    \begin{tabular}{cp{3.5cm}c}
 &  \hfil total time \hfil  & total space  \\ \hline\hline
\ELGIS & \hfil $\order{\sum_{S \in \sig S}\size{N[S]}}$ \hfil & $\order{\max_{S \in \sig S}\set{\size{N[S]}^3}}$ \\ \hline
\ELGS & \hfil $\order{\sum_{S \in \sig S}\size{V(G[S])}}$ \hfil & $\order{\max_{S \in \sig S}\set{\size{V(G[S])}^3}}$ \\ \hline 
    \end{tabular}
    \vspace{0.5em}
    \caption{Summary of our result. $\sig S$ is the set of all solutions. }
    \vspace{-1.5em}
    \label{tab:summary}
\end{table}

\bibliographystyle{abbrv}
\bibliography{main.bbl}

\begin{thebibliography}{10}

\bibitem{Alon:Hoory:GC:2002}
N.~Alon, S.~Hoory, and N.~Linial.
\newblock The moore bound for irregular graphs.
\newblock {\em Graphs and Combinatorics}, 18(1):53--57, 2002.

\bibitem{Bollobas:BOOK:2004}
B.~Bollob{\'a}s.
\newblock {\em Extremal graph theory}.
\newblock Courier Corporation, 2004.

\bibitem{Chandran:SIMADM:2003}
L.~S. Chandran.
\newblock A high girth graph construction.
\newblock {\em SIAM J. Discrete Math.}, 16(3):366--370, 2003.

\bibitem{Chang:Lu:SIAM:2013}
H.-C. Chang and H.-I. Lu.
\newblock Computing the girth of a planar graph in linear time.
\newblock {\em SIAM J. Comput.}, 42(3):1077--1094, 2013.

\bibitem{Alessio:Kante:COCOON:2017}
A.~Conte, M.~M. Kant{\'e}, Y.~Otachi, T.~Uno, and K.~Wasa.
\newblock Efficient enumeration of maximal k-degenerate subgraphs in a chordal
  graph.
\newblock In {\em Proc. {COCOON} 2017}, pages 150--161. Springer, 2017.

\bibitem{Alessio:Kazuhiro:COCOA:2017}
A.~Conte, K.~Kurita, K.~Wasa, and T.~Uno.
\newblock Listing acyclic subgraphs and subgraphs of bounded girth in directed
  graphs.
\newblock In {\em Proc. {COCOA} 2017}, volume 10628 of {\em LNCS}, pages
  169--181. Springer International Publishing, 2017.

\bibitem{Ferreira:Grossi:ESA:2011}
R.~Ferreira, R.~Grossi, and R.~Rizzi.
\newblock Output-sensitive listing of bounded-size trees in undirected graphs.
\newblock In {\em Proc. {ESA} 2011}, pages 275--286. Springer, 2011.

\bibitem{Itai:Rodeh:SIAM:1978}
A.~Itai and M.~Rodeh.
\newblock Finding a minimum circuit in a graph.
\newblock {\em SIAM J. Comput.}, 7(4):413--423, 1978.

\bibitem{Jonson:Yannakakis:IPL:1988}
D.~S. Johnson, M.~Yannakakis, and C.~H. Papadimitriou.
\newblock On generating all maximal independent sets.
\newblock {\em Inf. Process. Lett.}, 27(3):119 -- 123, 1988.

\bibitem{Kazuhiro:Kunihiro:arXiv:2018}
K.~Kurita, K.~Wasa, H.~Arimura, and T.~Uno.
\newblock Efficient enumeration of dominating sets for sparse graphs.
\newblock {\em arXiv preprint arXiv:1802.07863}, 2018.

\bibitem{Lazebnik:Ustimenko:BAMS:1995}
F.~Lazebnik, V.~A. Ustimenko, and A.~J. Woldar.
\newblock A new series of dense graphs of high girth.
\newblock {\em Bull. Am. Math. Soc.}, 32(1):73--79, 1995.

\bibitem{Parter:ICALP:2014}
M.~Parter.
\newblock Bypassing erd{\H{o}}s’ girth conjecture: hybrid stretch and
  sourcewise spanners.
\newblock In {\em Proc. {ICALP} 2014}, pages 608--619. Springer, 2014.

\bibitem{Read:Tarjan:Networks:1975}
R.~C. Read and R.~E. Tarjan.
\newblock Bounds on backtrack algorithms for listing cycles, paths, and
  spanning trees.
\newblock {\em Networks}, 3(5):237--252, 1975.

\bibitem{Shioura:Tamura:SICOMP:1997}
A.~Shioura, A.~Tamura, and T.~Uno.
\newblock {An Optimal Algorithm for Scanning All Spanning Trees of Undirected
  Graphs}.
\newblock {\em SIAM J. Comput.}, 26(3):678--692, 1997.

\bibitem{Wasa:Arimura:Uno:ISAAC:2014}
K.~Wasa, H.~Arimura, and T.~Uno.
\newblock {Efficient Enumeration of Induced Subtrees in a K-Degenerate Graph}.
\newblock In {\em Proc. {ISAAC} 2014}, volume 8889 of {\em LNCS}, pages
  94--102. Springer, 2014.

\end{thebibliography}

\end{document}